\documentclass[preprint,12pt]{elsarticle}

\usepackage{graphicx}
\usepackage{amsmath,amssymb,latexsym,amsfonts}
\usepackage{float}
\usepackage{subfig}
\usepackage{amsthm}
\newtheorem{theorem}{Theorem}
\newdefinition{definition}{Definition}
\newdefinition{remark}{Remark}
\newdefinition{comment}{Comment}
\newtheorem{corollary}{Corollary}
 \biboptions{sort&compress}
\allowdisplaybreaks
\journal{}

\begin{document}
\begin{frontmatter}



\title{Dual combination combination multi switching synchronization of eight chaotic systems}

\author[JMI]{Ayub Khan}
\author[KMC]{Dinesh Khattar}
\author[DOM]{Nitish Prajapati\corref{cor1}} 
\ead{nitishprajapati499@gmail.com}
\cortext[cor1]{Corresponding author}
\address[JMI]{Department of Mathematics, Jamia Millia Islamia, Delhi, India.}
\address[KMC]{Department of Mathematics, Kirorimal College, University of Delhi, Delhi, India.}
\address[DOM]{Department of Mathematics, University of Delhi, Delhi, India.}

\begin{abstract}
In this paper, a novel scheme for synchronizing four drive and four response systems is proposed by the authors. The idea of multi switching and dual combination synchronization is extended to dual combination-combination multi switching synchronization involving eight chaotic systems and is a first of its kind. Due to the multiple combination of chaotic systems and multi switching the resultant dynamic behaviour is so complex that, in communication theory, transmission and security of the resultant signal is more effective. Using Lyapunov stability theory, sufficient conditions are achieved and suitable controllers are designed to realise the desired synchronization. Corresponding theoretical analysis is presented and numerical simulations performed to demonstrate the effectiveness of the proposed scheme. 
\end{abstract}


\end{frontmatter}


\section{Introduction}
Much has been written and said about the concept of synchronization of chaotic systems since it was first introduced by Pecora and Caroll \cite{1}. Because of its interdisciplinary nature the chaos synchronization problem has received interest from researchers across the academic fields such as physics, mathematics, engineering, biology, chemistry, etc. The potential applications of chaos synchronization to engineering systems, information processing, secure communications, and biomedical science amongst many others has led to a vast variety of research studies in this topic of nonlinear science \cite{2, 3, 4, 5}. Various kinds of synchronization such as complete synchronization, anti synchronization, projective synchronization, reduced order synchronization, etc. have been reported and presented in a variety of chaotic systems using many effective methods such as active control, adaptive control, backstepping control, sliding mode control and so on \cite{6, 7, 8, 9, 10, 11, 12, 13}. 

Amongst many synchronization schemes, dual synchronization of chaotic systems is one which has successfully piqued the scientific curiosity of researchers because of its challenging and non traditional nature. Deviating from the traditional approach of synchronizing one drive and one response system, in dual synchronization two drive systems are synchronized with two response systems. Since the first inception of the idea by Liu and Davids \cite{14}, it has been extensively investigated in various synchronization studies \cite{15, 16, 17}. In all prior work the common theme is to consider one pair of drive system with one pair of response system. Only recently the idea of dual synchronization was extended to two pair of drive systems and one pair of response system \cite{18, 19}.

Synchronization studies involving multiple drive and response systems is a relatively unexplored area of research. New ideas have recently been initiated in the study of chaos synchronization where multiple chaotic systems are involved. In the literature of chaos synchronization the addition of combination synchronization \cite{20, 21}, combination combination synchronization \cite{22, 23}, compound synchronization \cite{24, 25}, double compound synchronization \cite{26}, compound combination synchronization \cite{27, 28} etc. has opened new research directions to be explored. These significant ideas have strengthened the security of information transmission because of the complexity which they bring in transmitted signals.

Multi switching synchronization of chaotic systems is yet another relatively unexplored area of research \cite{29}. In this non conventional scheme, different states of the drive system are synchronized with different state of the response system. Due to this, a wide range of synchronization direction exists for multi switching synchronization schemes. The importance of such kind of studies to information security cannot be emphasised enough and thus makes them a very relevant topic to be investigated. A few reported work in this direction can be studied in \cite{30, 31, 32, 33}. To the knowledge of authors the diverse possibilities of multi switching synchronization have not yet been explored with the dual synchronization schemes.

In this paper, motivated by the above discussion, the authors have combined the idea of multi switching with dual synchronization and extended it to combination combination synchronization of four chaotic systems. The novel scheme, dual combination combination multi switching synchronization involves eight chaotic systems of which four are drive systems and four are response systems. In contrast to double compound synchronization involving four drive and two response systems, this work is a significant improvement and extension. Using Lyapunov stability theory, sufficient conditions have been achieved to realise the desired synchronization. To demonstrate the effectiveness of the proposed method numerical simulations have been performed. The main contribution and advantages of this study are: a) The dual  synchronization study has been extended to pair of two drive and two response systems to achieve dual combination combination synchronization. no such work involving two pair of response systems has earlier been reported. b) The multi switching synchronization scheme is combined with dual combination combination synchronization to achieve the novel scheme which is a first of its kind. No previous work on dual multi switching studies exist. c) The proposed scheme successfully synchronizes eight chaotic systems of which four are drive and four are response systems. The complexity of signals due to multiple combination and the number of synchronization directions due to multi switching vastly enhances the anti attack ability of any signal that will be transmitted using combination of two pair of drive systems. d) Several existing synchronization schemes are obtained as special cases of dual combination combination multi switching synchronization.

\section{Formulation of dual combination combination multi switching synchronization}
In this section, we formulate the scheme of DCCMS of chaotic systems. We require two pair of four chaotic drive systems and two pair of four chaotic response systems. Let the first two drive systems be described as 
\begin{equation} \label{1} 
\dot{x_1}=f_1(x_1)
\end{equation}
\begin{equation} \label{2}
\dot{x_2}=f_2(x_2) 
\end{equation} 
where $x_1=(x_{11}, x_{12}, ..., x_{1n})^T$, $x_2=(x_{21}, x_{22}, ..., x_{2n})^T$, $f_1$, $f_2$: $R^n \rightarrow R^n$ are known continuous vector functions. Linear combination of the states of two drive systems (\ref{1}) and (\ref{2}) gives a resultant signal of the form
\begin{equation} \label{3}
\begin{aligned}
S_1=&[a_{11}x_{11}, a_{12}x_{12}, ..., a_{1n}x_{1n}, a_{21}x_{21}, a_{22}x_{22}, ..., a_{2n}x_{2n}]^T \\
      =&\begin{bmatrix}
             A_1 & 0\\
             0 & A_2\\
          \end{bmatrix}  
          \begin{bmatrix}
             x_1\\
             x_2\\
          \end{bmatrix} 
       = Ax   
\end{aligned}                 
\end{equation}
where $A_1=diag(a_{11}, a_{12}, ..., a_{1n})$, and $A_2=diag(a_{21}, a_{22}, ..., a_{2n})$ are two known matrices and $a_{1i}, a_{2j}$ are not all zero at the same time $(i, j = 1, 2, ..., n)$.

Next two drive systems are written as
\begin{equation} \label{4} 
\dot{y_1}=g_1(y_1)
\end{equation}
\begin{equation} \label{5}
\dot{y_2}=g_2(y_2) 
\end{equation} 
where $y_1=(y_{11}, y_{12}, ..., y_{1n})^T$, $y_2=(y_{21}, y_{22}, ..., y_{2n})^T$, $g_1$, $g_2$: $R^n \rightarrow R^n$ are known continuous vector functions. Hence, the linear combination of the states of two drive systems (\ref{4}) and (\ref{5}) gives a resultant signal of the form
\begin{equation} \label{6}
\begin{aligned}
S_2=&[b_{11}y_{11}, b_{12}y_{12}, ..., b_{1n}y_{1n}, b_{21}y_{21}, b_{22}y_{22}, ..., b_{2n}y_{2n}]^T \\
      =&\begin{bmatrix}
             B_1 & 0\\
             0 & B_2\\
          \end{bmatrix}  
          \begin{bmatrix}
             y_1\\
             y_2\\
          \end{bmatrix} 
       = By   
\end{aligned}                 
\end{equation}
where $B_1=diag(b_{11}, b_{12}, ..., b_{1n})$, and $B_2=diag(b_{21}, b_{22}, ..., b_{2n})$ are two known matrices and $b_{1i}, b_{2j}$ are not all zero at the same time $(i, j = 1, 2, ..., n)$.

Let the first two response systems be given by 
\begin{equation} \label{7} 
\dot{z_1}=h_1(z_1)+u_1
\end{equation}
\begin{equation} \label{8}
\dot{z_2}=h_2(z_2)+u_2 
\end{equation} 
where $z_1=(z_{11}, z_{12}, ..., z_{1n})^T$, $z_2=(z_{21}, z_{22}, ..., z_{2n})^T$, $h_1$, $h_2$: $R^n \rightarrow R^n$ are known continuous vector functions, and $u_1=(u_{11}, u_{12}, ..., u_{1n})$, $u_2=(u_{21}, u_{22}, ..., u_{2n})$ are the controllers to be designed. By linear combination of the states of two response systems (\ref{7}) and (\ref{8}) a resultant signal is obtained of the form
\begin{equation} \label{9}
\begin{aligned}
S_3=&[c_{11}z_{11}, c_{12}z_{12}, ..., c_{1n}z_{1n}, c_{21}z_{21}, c_{22}z_{22}, ..., c_{2n}z_{2n}]^T \\
      =&\begin{bmatrix}
             C_1 & 0\\
             0 & C_2\\
          \end{bmatrix}  
          \begin{bmatrix}
             z_1\\
             z_2\\
          \end{bmatrix} 
       = Cz   
\end{aligned}                 
\end{equation}
where $C_1=diag(c_{11}, c_{12}, ..., c_{1n})$, and $C_2=diag(c_{21}, c_{22}, ..., c_{2n})$ are two known matrices and $c_{1i}, c_{2j}$ are not all zero simultaneously $(i, j = 1, 2, ..., n)$.

Let the next two response systems be described as 
\begin{equation} \label{10} 
\dot{w_1}=k_1(w_1)+u_3
\end{equation}
\begin{equation} \label{11}
\dot{w_2}=k_2(w_2)+u_4 
\end{equation} 
where $w_1=(w_{11}, w_{12}, ..., w_{1n})^T$, $w_2=(w_{21}, w_{22}, ..., w_{2n})^T$, $k_1$, $k_2$: $R^n \rightarrow R^n$ are known continuous vector functions, and $u_3=(u_{31}, u_{32}, ..., u_{3n})$, $u_4=(u_{41}, u_{42}, ..., u_{4n})$ are the controllers to be designed. Linear combination of the states of two response systems (\ref{10}) and (\ref{11}) gives a resultant signal of the form
\begin{equation} \label{12}
\begin{aligned}
S_4=&[d_{11}z_{11}, d_{12}z_{12}, ..., d_{1n}z_{1n}, d_{21}z_{21}, d_{22}z_{22}, ..., d_{2n}z_{2n}]^T \\
      =&\begin{bmatrix}
             D_1 & 0\\
             0 & D_2\\
          \end{bmatrix}  
          \begin{bmatrix}
             z_1\\
             z_2\\
          \end{bmatrix} 
       = Dz   
\end{aligned}                 
\end{equation}
where $D_1=diag(d_{11}, d_{12}, ..., d_{1n})$, and $D_2=diag(d_{21}, d_{22}, ..., d_{2n})$ are two known matrices and $d_{1i}, d_{2j}$ are not all zero at the same time $(i, j = 1, 2, ..., n)$.

The error signal for dual combination combination synchronization is 
\begin{equation} \label{13}
\begin{aligned}
e=&S_1+S_2-S_3-S_4 \\
  =& Ax+By-Cz-Dw \\
  =&\begin{bmatrix}
             A_1 & 0\\
             0 & A_2\\
          \end{bmatrix}  
          \begin{bmatrix}
             x_1\\
             x_2\\
          \end{bmatrix}
      + \begin{bmatrix}
             B_1 & 0\\
             0 & B_2\\
          \end{bmatrix}  
          \begin{bmatrix}
             y_1\\
             y_2\\
          \end{bmatrix}
        - \begin{bmatrix}
             C_1 & 0\\
             0 & C_2\\
          \end{bmatrix}  
          \begin{bmatrix}
             z_1\\
             z_2\\
          \end{bmatrix}
        - \begin{bmatrix}
             D_1 & 0\\
             0 & D_2\\
          \end{bmatrix}  
          \begin{bmatrix}
             w_1\\
             w_2\\
          \end{bmatrix}\\
     =&\begin{bmatrix}
          A_1x_1+B_1y_1-C_1z_1-D_1w_1\\
          A_2x_2+B_2y_2-C_2z_2-D_2w_2\\
          \end{bmatrix}
\end{aligned}
\end{equation}
        
\begin{definition}
If there exist four constant diagonal matrices $A, B, C, D \in R^{2n \times 2n}$ and $C \neq 0$ or $D \neq 0$ such that
\begin{equation} \label{14}
\lim_{t \to \infty} \| e \| = \lim_{t \to \infty} \| Ax+By-Cz-Dw \| = 0,
\end{equation} 
where $\|.\|$ is the vector norm, then the drive systems (\ref{1}), (\ref{2}), (\ref{4}), and (\ref{5}) realise dual combination combination synchronization with the response systems (\ref{7}), (\ref{8}), (\ref{10}), and (\ref{11}).
\end{definition}
\begin{remark}
The diagonal matrices $A$, $B$, $C$, and $D$ are called the scaling matrices and can be extended to functional matrices of state variables $x$, $y$, $z$, and $w$. 
\end{remark} 
\begin{definition}
If the indices of the error states $e_{1m_{(ijlm)}}$, and $e_{2m_{(ijlm)}}$ are redefined such that $i=j=l \neq m$ or $i=j=m \neq l$ or $i=l=m \neq j$ or $j=l=m \neq i$; or $i=j \neq l=m$ or $i=l \neq j=m$ or $i=m \neq j=l$; or $i=j \neq l \neq m$ or $i=l \neq j \neq m$ or $i=m \neq l \neq j$ or $i \neq j=l \neq m$ or $i \neq j \neq l=m$ or $i \neq l \neq j=m$; or $i \neq j \neq l \neq m$ and 
\begin{equation} \label{16}
\lim_{t \to \infty}  \| e \| = \lim_{t \to \infty} \| Ax+By-Cz-Dw \| = 0, 
\end{equation}  
where $i, j, l, m = 1, 2, ..., n$ and $\|.\|$ is the vector norm, then the drive systems (\ref{1}), (\ref{2}), (\ref{4}), and (\ref{5}) are said to be in dual combination combination multi switching synchronization with response systems (\ref{7}), (\ref{8}), (\ref{10}), and (\ref{11}).
\end{definition}
\begin{remark}
If $A_1=B_1=C_1=D_1=0$, or $A_2=B_2=C_2=D_2=0$, then dual combination combination multi switching synchronization changes to multi switching combination combination synchronization problem of chaotic systems.
\end{remark}
\begin{remark}
If $C_1=C_2=0$, or $D_1=D_2=0$, then dual combination combination multi switching synchronization changes to dual combination multi switching synchronization of chaotic systems.
\end{remark}
\begin{remark}
If $A_1=B_1=C_1=D_1=0$, and $C_2=0$ or $D_2=0$, or $A_2=B_2=C_2=D_2=0$, and $C_1=0$ or $D_1=0$, then dual combination combination multi switching synchronization changes to multi switching combination synchronization of chaotic systems.
\end{remark}
\begin{remark}
Using suitable values for the scaling factors $A_1$, $A_2$, $B_1$, $B_2$, $C_1$, $C_2$, $D_1$, and $D_2$, multi switching dual projective synchronization and multi switching projective synchronization may also be obtained by the proposed scheme.
\end{remark}
\section{Synchronization Theory}
To achieve the dual combination combination multi switching synchronization among four chaotic drive systems and four chaotic response systems, let the control functions be defined as
\begin{equation} \label{17}
\left\{
\begin{aligned}
U_{1m}=a_{1i}f_{1i}+b_{1j}g_{1j}-c_{1l}h_{1l}-d_{1m}k_{1m}+e_{1m_{(ijlm)}}, \quad{(i, j, l, m=1, 2, ..., n)}\\
U_{2m}=a_{2i}f_{2i}+b_{2j}g_{2j}-c_{2l}h_{2l}-d_{2m}k_{2m}+e_{2m_{(ijlm)}}, \quad{(i, j, l, m=1, 2, ..., n)}\\
\end{aligned}
\right.
\end{equation}
where 
\begin{equation} \label{18}
\left\{
\begin{aligned}
U_{1m}=c_{1l}u_{1l}+d_{1m}u_{3m}, \quad{(l, m=1, 2, ..., n)}\\
U_{2m}=c_{2l}u_{2l}+d_{2m}u_{4m}, \quad{(l, m=1, 2, ..., n)}\\
\end{aligned}
\right.
\end{equation}
and $f_1=(f_{11}, f_{12}, ..., f_{1n})^T$, $f_2=(f_{21}, f_{22}, ..., f_{2n})^T$, $g_1=(g_{11}, g_{12}, ..., g_{1n})^T$, $g_2=(g_{21}, g_{22}, ..., g_{2n})^T$, $h_1=(h_{11}, h_{12}, ..., h_{1n})^T$, $h_2=(h_{21}, h_{22}, ..., h_{2n})^T$, $k_1=(k_{11}, k_{12}, ..., k_{1n})^T$, and $k_2=(k_{21}, k_{22}, ..., k_{2n})^T$. 
\begin{theorem}
The drive systems (\ref{1}), (\ref{2}), (\ref{4}), and (\ref{5}) achieve dual combination combination multi switching synchronization with response systems (\ref{7}), (\ref{8}), (\ref{10}), and (\ref{11}) if the control functions are chosen as given in (\ref{17}).
\end{theorem}
\begin{proof}
Using (\ref{13}) the error dynamical system can be written as 
\begin{equation} \label{19}
\dot{e}=\begin{bmatrix}
             \dot{e}_1\\
             \dot{e}_2\\
          \end{bmatrix}  
       = \begin{bmatrix}
          A_1\dot{x}_1+B_1\dot{y}_1-C_1\dot{z}_1-D_1\dot{w}_1\\
          A_2\dot{x}_2+B_2\dot{y}_2-C_2\dot{z}_2-D_2\dot{w}_2\\
          \end{bmatrix}
\end{equation}
which can be further written as
\begin{equation} \label{20}
\begin{bmatrix}
             \dot{e}_1\\
             \dot{e}_2\\
          \end{bmatrix}  
       = \begin{bmatrix}
          A_1f_1+B_1g_1-C_1(h_1+u_1)-D_1(k_1+u_3)\\
          A_2f_2+B_2g_2-C_2(h_2+u_2)-D_2(k_2+u_4)\\
          \end{bmatrix}
\end{equation}
From this we obtain
\begin{equation} \label{21}
\left\{
\begin{aligned}
\dot{e}_{1m_{(ijlm)}}=&a_{1i}f_{1i}+b_{1j}g_{1j}-c_{1l}(h_{1l}+u_{1l})-d_{1m}(k_{1m}+u_{3m}),\\
 &\quad{(i, j, l, m=1, 2, ..., n)}\\
\dot{e}_{2m_{(ijlm)}}=&a_{2i}f_{2i}+b_{2j}g_{2j}-c_{2l}(h_{2l}+u_{2l})-d_{2m}(k_{2m}+u_{4m}),\\ 
&\quad{(i, j, l, m=1, 2, ..., n)}\\
\end{aligned}
\right.
\end{equation}
where the indices (ijlm) satisfies one of the generic conditions given in Definition 2.

\noindent Let the Lyapunov function be defined as
\begin{align*}
V=&\frac{1}{2}e^Te\\
   =&\frac{1}{2} \sum_{m=1}^{n} (e_{1m_{(ijlm)}})^2 + \frac{1}{2} \sum_{m=1}^{n} (e_{2m_{(ijlm)}})^2 
\end{align*}
The derivative $\dot{V}$ is obtained as
\begin{equation} 
\dot{V}=\sum_{m=1}^{n} e_{1m_{(ijlm)}}\dot{e}_{1m_{(ijlm)}}+\sum_{m=2}^{n} e_{2m_{(ijlm)}}\dot{e}_{2m_{(ijlm)}}
\end{equation}
Using (\ref{17}) and (\ref{21}) in the above equation we get
\begin{align*}
\dot{V}=&\sum_{m=1}^{n} e_{1m_{(ijlm)}}[a_{1i}f_{1i}+b_{1j}g_{1j}-c_{1l}h_{1l}-d_{1m}k_{1m}-U_{1m}]\\
            &+\sum_{m=1}^{n} e_{2m_{(ijlm)}}[a_{2i}f_{2i}+b_{2j}g_{2j}-c_{2l}h_{2l}-d_{2m}k_{2m}-U_{2m}]\\
            =&\sum_{m=1}^{n} e_{1m_{(ijlm)}}(-e_{1m_{(ijlm)}})+\sum_{m=1}^{n} e_{2m_{(ijlm)}}(-e_{2m_{(ijlm)}}) \quad \text (Using (\ref{17}))\\
            =&-e^Te
\end{align*}
Thus we see that $\dot{V}$ is negative definite. Using Lyapunov stability theory, we get $\lim_{t \to \infty} \| e \|=0$, which gives us $\lim_{t \to \infty} \| e_1 \|=0$ and $\lim_{t \to \infty} \| e_2 \|=0$. This means that the drive systems (\ref{1}), (\ref{2}), (\ref{4}), and (\ref{5}) achieve dual combination combination multi switching synchronization with response systems (\ref{7}), (\ref{8}), (\ref{10}), and (\ref{11}). 
\end{proof}

The following corollaries are easily obtained from Theorem 1 and their proofs are omitted here.
\begin{corollary}
(i) If $a_{1i}=b_{1j}=c_{1l}=d_{1m}=0$, $i, j, l, m= 1, 2, ..., n$ then the drive systems (\ref{2}) and (\ref{5}) achieve multi switching combination combination synchronization with the response systems (\ref{8}) and (\ref{11}) provided the control function is chosen as
\begin{equation*}
U_{2m}=a_{2i}f_{2i}+b_{2j}g_{2j}-c_{2l}h_{2l}-d_{2m}k_{2m}+e_{2m_{(ijlm)}}, \quad{(i, j, l, m=1, 2, ..., n)}\\
\end{equation*}
(ii)  If $a_{2i}=b_{2j}=c_{2l}=d_{2m}=0$, $i, j, l, m= 1, 2, ..., n$ then the drive systems (\ref{1}) and (\ref{4}) achieve multi switching combination combination synchronization with the response systems (\ref{7}) and (\ref{10}) provided the control function is chosen as
\begin{equation*}
U_{2m}=a_{2i}f_{2i}+b_{2j}g_{2j}-c_{2l}h_{2l}-d_{2m}k_{2m}+e_{2m_{(ijlm)}}, \quad{(i, j, l, m=1, 2, ..., n)}\\
\end{equation*}
\end{corollary}
\begin{corollary}
(i) If $c_{1l}=c_{2l}=0$, $l= 1, 2, ..., n$ then the drive systems (\ref{1}), (\ref{2}), (\ref{4}) and (\ref{5}) achieve dual combination multi switching synchronization with the response systems (\ref{10}) and (\ref{11}) provided the control functions are chosen as
\begin{align*}
u_{3m}=&d_{1m}^{-1}a_{1i}f_{1i}+d_{1m}^{-1}b_{1j}g_{1j}-k_{1m}+d_{1m}^{-1}e_{1m_{(ijlm)}}, \quad{(i, j, l, m=1, 2, ..., n)}\\
u_{4m}=&d_{2m}^{-1}a_{2i}f_{2i}+d_{2m}^{-1}b_{2j}g_{2j}-k_{2m}+d_{2m}^{-1}e_{2m_{(ijlm)}}, \quad{(i, j, l, m=1, 2, ..., n)}\\
\end{align*}
(ii) If $d_{1m}=d_{2m}=0$, $m= 1, 2, ..., n$ then the drive systems (\ref{1}), (\ref{2}), (\ref{4}) and (\ref{5}) achieve dual combination multi switching synchronization with the response systems (\ref{7}) and (\ref{8}) provided the control functions are chosen as
\begin{align*}
u_{1l}=&c_{1l}^{-1}a_{1i}f_{1i}+c_{1l}^{-1}b_{1j}g_{1j}-h_{1l}+c_{1l}^{-1}e_{1m_{(ijlm)}},\quad{(i, j, l, m=1, 2, ..., n)}\\
u_{2l}=&c_{2l}^{-1}a_{2i}f_{2i}+c_{2l}^{-1}b_{2j}g_{2j}-h_{2l}+c_{2l}^{-1}e_{2m_{(ijlm)}}, \quad{(i, j, l, m=1, 2, ..., n)}\\
\end{align*} 
\end{corollary}
\begin{corollary}
(i) If $a_{1i}=b_{1j}=c_{1l}=d_{1m}=0$, and $c_{2m}=0$, $i, j, l, m= 1, 2, ..., n$ then the drive systems (\ref{2}) and (\ref{5}) achieve multi switching combination synchronization with the response system (\ref{11}) provided the control function is chosen as
\begin{equation*}
u_{4m}=d_{2m}^{-1}a_{2i}f_{2i}+d_{2m}^{-1}b_{2j}g_{2j}-k_{2m}+d_{2m}^{-1}e_{2m_{(ijlm)}}, \quad{(i, j, l, m=1, 2, ..., n)}\\
\end{equation*}
(ii) If $a_{1i}=b_{1j}=c_{1l}=d_{1m}=0$, and $d_{2m}=0$, $i, j, l, m= 1, 2, ..., n$ then the drive systems (\ref{2}) and (\ref{5}) achieve multi switching combination synchronization with the response system (\ref{8}) provided the control function is chosen as
\begin{equation*}
u_{2l}=c_{2l}^{-1}a_{2i}f_{2i}+c_{2l}^{-1}b_{2j}g_{2j}-h_{2l}+c_{2l}^{-1}e_{2m_{(ijlm)}}, \quad{(i, j, l, m=1, 2, ..., n)}\\
\end{equation*}
(iii) If $a_{2i}=b_{2j}=c_{2l}=d_{2m}=0$, and $c_{1m}=0$, $i, j, l, m= 1, 2, ..., n$ then the drive systems (\ref{1}) and (\ref{4}) achieve multi switching combination synchronization with the response system (\ref{10}) provided the control function is chosen as
\begin{equation*}
u_{3m}=d_{1m}^{-1}a_{1i}f_{1i}+d_{1m}^{-1}b_{1j}g_{1j}-k_{1m}+d_{1m}^{-1}e_{1m_{(ijlm)}}, \quad{(i, j, l, m=1, 2, ..., n)}\\
\end{equation*}
(iv) If $a_{2i}=b_{2j}=c_{2l}=d_{2m}=0$, and $d_{1m}=0$, $i, j, l, m= 1, 2, ..., n$ then the drive systems (\ref{1}) and (\ref{4}) achieve multi switching combination synchronization with the response system (\ref{7}) provided the control function is chosen as
\begin{equation*}
u_{1l}=c_{1l}^{-1}a_{1i}f_{1i}+c_{1l}^{-1}b_{1j}g_{1j}-h_{1l}+c_{1l}^{-1}e_{1m_{(ijlm)}}, \quad{(i, j, l, m=1, 2, ..., n)}\\
\end{equation*}
\end{corollary}
\section{Illustration of the synchronization scheme}
In this section we realize the dual combination combination multi switching synchronization among eight chaotic systems and perform numerical simulations  to show the validity and effectiveness of the proposed scheme. As an example we consider Genesio Tesi system and Lu system to demonstrate the method. let the first two drive systems be given as
\begin{equation} \label{23}
\left\{
\begin{aligned}
\dot{x}_{11}=&x_{12}\\
\dot{x}_{12}=&x_{13}\\
\dot{x}_{13}=&-6x_{11}-2.92x_{12}-1.2x_{13}+x_{11}^2
\end{aligned}
\right.
\end{equation}

\begin{equation} \label{24}
\left\{
\begin{aligned}
\dot{x}_{21}=&36(x_{22}-x_{21})\\
\dot{x}_{22}=&-x_{21}x_{23}+20x_{22}\\
\dot{x}_{23}=&x_{21}x_{22}-3x_{23}
\end{aligned}
\right.
\end{equation}
The next two drive systems are considered as
\begin{equation} \label{25}
\left\{
\begin{aligned}
\dot{y}_{11}=&y_{12}\\
\dot{y}_{12}=&y_{13}\\
\dot{y}_{13}=&-6y_{11}-2.92y_{12}-1.2y_{13}+y_{11}^2
\end{aligned}
\right.
\end{equation}

\begin{equation} \label{26}
\left\{
\begin{aligned}
\dot{y}_{21}=&36(y_{22}-y_{21})\\
\dot{y}_{22}=&-y_{21}y_{23}+20y_{22}\\
\dot{y}_{23}=&y_{21}y_{22}-3y_{23}
\end{aligned}
\right.
\end{equation}
The first two response system are described as
\begin{equation} \label{27}
\left\{
\begin{aligned}
\dot{z}_{11}=&z_{12}+u_{11}\\
\dot{z}_{12}=&z_{13}+u_{12}\\
\dot{z}_{13}=&-6z_{11}-2.92z_{12}-1.2z_{13}+z_{11}^2+u_{13}
\end{aligned}
\right.
\end{equation}

\begin{equation} \label{28}
\left\{
\begin{aligned}
\dot{z}_{21}=&36(z_{22}-z_{21})+u_{21}\\
\dot{z}_{22}=&-z_{21}z_{23}+20z_{22}+u_{22}\\
\dot{z}_{23}=&z_{21}z_{22}-3z_{23}+u_{23}
\end{aligned}
\right.
\end{equation}
and the next two response systems are taken as
\begin{equation} \label{29}
\left\{
\begin{aligned}
\dot{w}_{11}=&w_{12}+u_{31}\\
\dot{w}_{12}=&w_{13}+u_{32}\\
\dot{w}_{13}=&-6w_{11}-2.92w_{12}-1.2w_{13}+w_{11}^2+u_{33}
\end{aligned}
\right.
\end{equation}

\begin{equation} \label{30}
\left\{
\begin{aligned}
\dot{w}_{21}=&36(w_{22}-w_{21})+u_{41}\\
\dot{w}_{22}=&-w_{21}w_{23}+20w_{22}+u_{42}\\
\dot{w}_{23}=&w_{21}w_{22}-3w_{23}+u_{43}
\end{aligned}
\right.
\end{equation}
By the conditions on indices $i, j, l, m=1, 2, 3$ stated in Definition 2, several multi switching combination exist for defining the error $e=(e_1, e_2)^T$. We will present results for one randomly selected error space vector combination formed out of several possibilities. Let us define $e_1=(e_{11_{2131}}, e_{12_{1322}}, e_{13_{3213}})$, and $e_2=(e_{21_{3221}}, e_{22_{1332}}, e_{23_{2113}})$ where
\begin{align*}
e_{11_{2131}}=&a_{12}x_{12}+b_{11}y_{11}-c_{13}z_{13}-d_{11}w_{11}\\
e_{12_{1322}}=&a_{11}x_{11}+b_{13}y_{13}-c_{12}z_{12}-d_{12}w_{12}\\
e_{13_{3213}}=&a_{13}x_{13}+b_{12}y_{12}-c_{11}z_{11}-d_{13}w_{13}\\
e_{21_{3221}}=&a_{23}x_{23}+b_{22}y_{22}-c_{22}z_{22}-d_{21}w_{21}\\
e_{22_{1332}}=&a_{21}x_{21}+b_{23}y_{23}-c_{23}z_{23}-d_{22}w_{22}\\
e_{23_{2113}}=&a_{22}x_{22}+b_{21}y_{21}-c_{21}z_{21}-d_{23}w_{23}
\end{align*}
under the assumption that $A_1=diag(a_{11}, a_{12}, a_{13})$, $A_2=diag(a_{21}, a_{22}, a_{23})$, $B_1=diag(b_{11}, b_{12}, b_{13})$, $B_2=diag(b_{21}, b_{22}, b_{23})$, $C_1=diag(c_{11}, c_{12}, c_{13})$, $C_2=diag(c_{21}, c_{22}, c_{23})$, $D_1=diag(d_{11}, d_{12}, d_{13})$, and $D_2=diag(d_{21}, d_{22}, d_{23})$.
Assuming $A_1=A_2=B_1=B_2=C_1=C_2=D_1=D_2=I$, the controllers are chosen as
\begin{equation} \label{31}
\left\{
\begin{aligned}
U_{11}=&x_{13}+y_{12}-(-6z_{11}-2.92z_{12}-1.2z_{13}+z_{11}^2)-w_{12}+e_{11_{2131}}\\
U_{12}=&x_{12}+(-6y_{11}-2.92y_{12}-1.2y_{13}+y_{11}^2)-z_{13}-w_{13}+e_{12_{1322}}\\
U_{13}=&(-6x_{11}-2.92x_{12}-1.2x_{13}+x_{11}^2)+y_{13}-z_{12}\\
             &-(-6w_{11}-2.92w_{12}-1.2w_{13}+w_{11}^2)+e_{13_{3213}}\\
\end{aligned}
\right.
\end{equation} 

\begin{equation} \label{32}
\left\{
\begin{aligned}
U_{21}=&x_{21}x_{22}-3x_{23}+(-y_{21}y_{23}+20y_{22})-(-z_{21}z_{23}+20z_{22})-36(w_{22}-w_{21})+e_{21_{3221}}\\
U_{22}=&36(x_{22}-x_{21})+(y_{21}y_{22}-3y_{23})-(z_{21}z_{22}-3z_{23})-(-w_{21}w_{23}+20w_{22})+e_{22_{1332}}\\
U_{23}=&(-x_{21}x_{23}+20x_{22})+36(y_{22}-y_{21})-36(z_{22}-z_{21})-(w_{21}w_{22}-3w_{23})+e_{23_{2113}}
\end{aligned}
\right.
\end{equation} 
where $U_{11}=u_{13}+u_{31}$, $U_{12}=u_{12}+u_{32}$, $U_{13}=u_{11}+u_{33}$, $U_{21}=u_{22}+u_{41}$, $U_{22}=u_{23}+u_{42}$, and $U_{23}=u_{21}+u_{43}$.

These controllers (\ref{31}) and (\ref{32}) are designed in accordance with Theorem 1 in order to realise the desired synchronization. In the numerical simulations process the initial conditions of the drive and response systems are chosen as $(x_{11}, x_{12}, x_{13})=(2, -3, 1)$, $(x_{21}, x_{22}, x_{23})=(-2.5, 1, -3)$, $(y_{11}, y_{12}, y_{13})=(1, 0, -1)$, $(y_{21}, y_{22}, y_{23})=(-1.5, 2, 1.5)$, $(z_{11}, z_{12}, z_{13})=(4, -3.5, 3)$, $(z_{21}, z_{22}, z_{23})=(-0.5, 1.5, 0)$, $(w_{11}, w_{12}, w_{13})=(1, -1.5, -2)$, and $(w_{21}, w_{22}, w_{23})=(-1, 1.5, 3)$. Figures $(1)-(6)$ illustrates the time response of synchronized states and Figure $(7)$ displays time response of synchronization errors. We can see that the desired dual combination combination multi switching synchronization is achieved with the controllers we designed.  
\begin{figure}[htbp] 
   \centering
   \includegraphics[width=\linewidth,height=1.5 in]{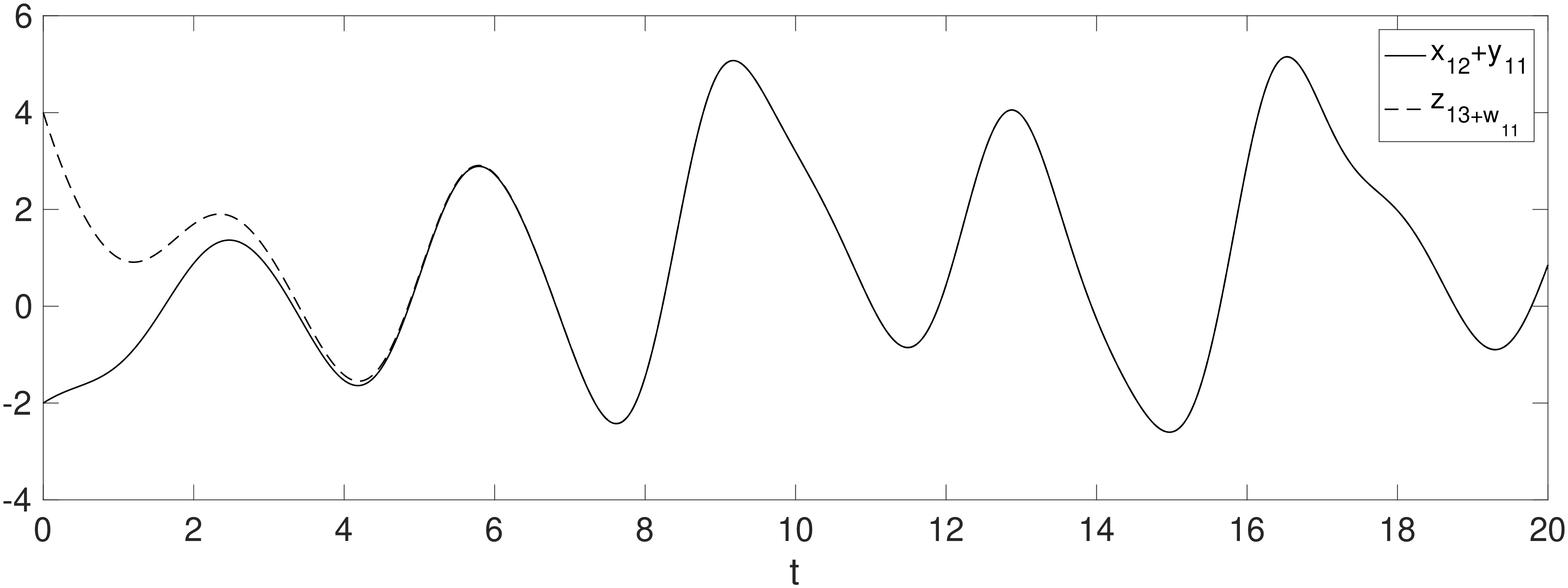} 
   \caption{Response for states $x_{12}+y_{11}$ and $z_{13}+w_{11}$ for drive systems (\ref{23}), (\ref{25}) and response systems (\ref{27}), (\ref{29}).}
\end{figure}
\begin{figure}[tbp] 
   \centering
   \includegraphics[width=\linewidth,height=1.5in]{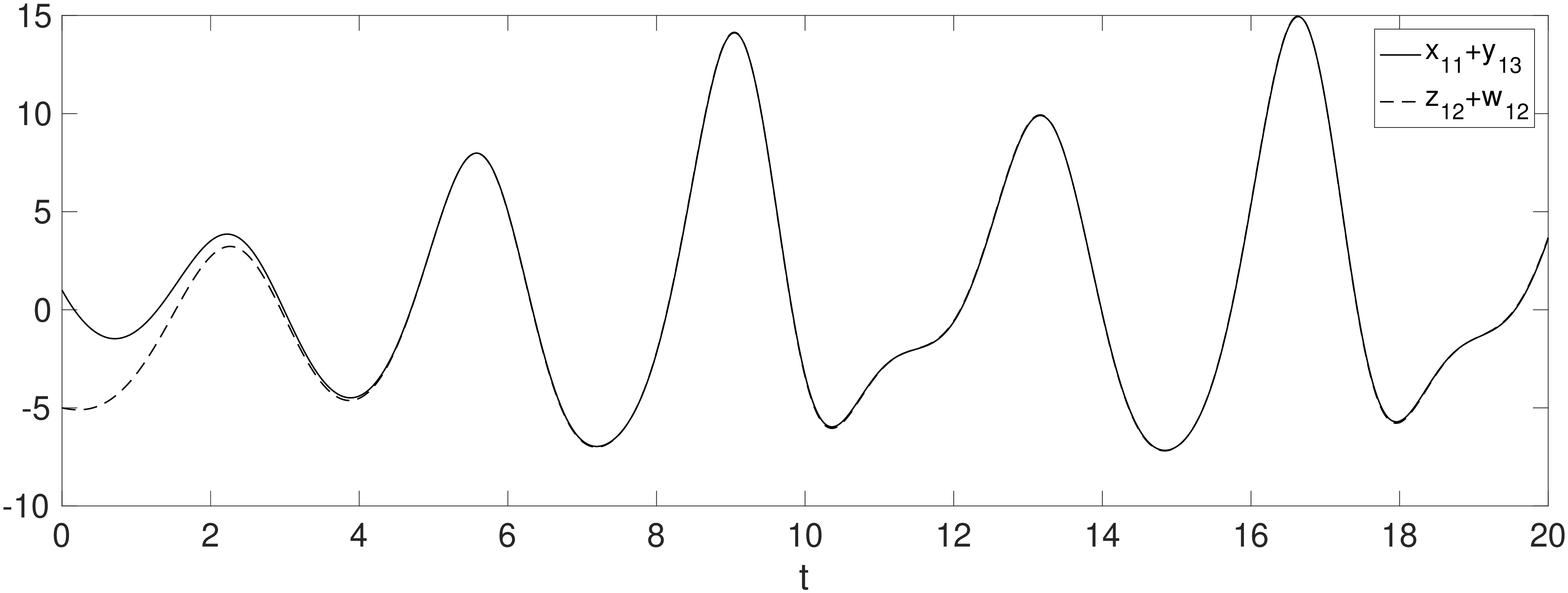} 
   \caption{Response for states $x_{11}+y_{13}$ and $z_{12}+w_{12}$ for drive systems (\ref{23}), (\ref{25}) and response systems (\ref{27}), (\ref{29}).}
\end{figure}
\begin{figure}[htbp] 
   \centering
   \includegraphics[width=\linewidth,height=1.5in]{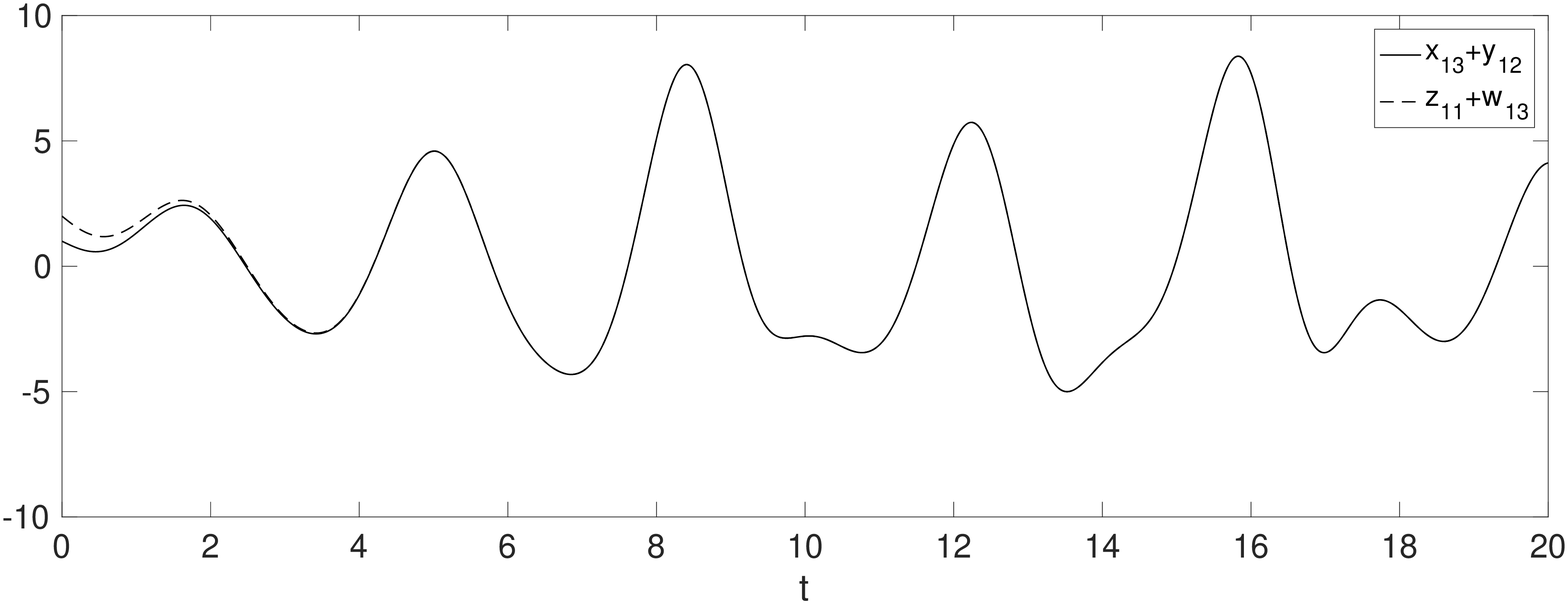} 
   \caption{Response for states $x_{13}+y_{12}$ and $z_{11}+w_{13}$ for drive systems (\ref{23}), (\ref{25}) and response systems (\ref{27}), (\ref{29}).}
\end{figure}
\begin{figure}[htbp] 
   \centering
   \includegraphics[width=\linewidth,height=1.5in]{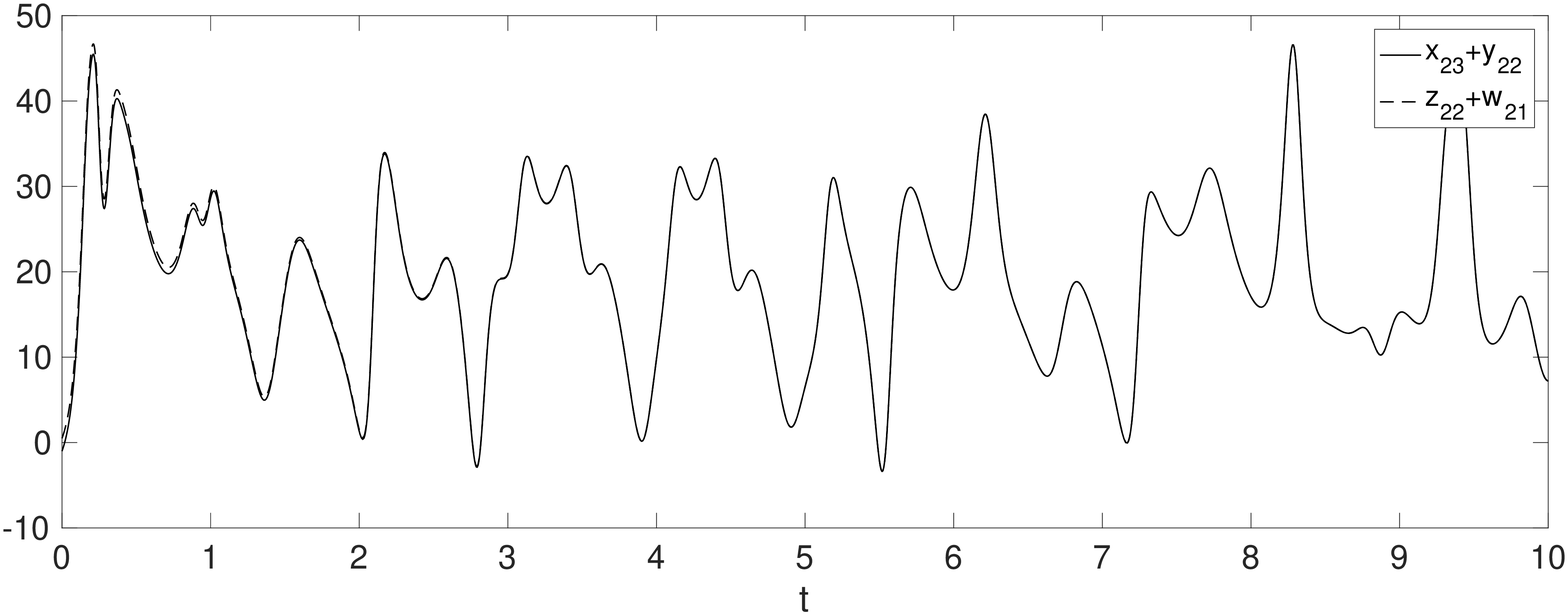} 
   \caption{Response for states $x_{23}+y_{22}$ and $z_{22}+w_{21}$ for drive systems (\ref{24}), (\ref{26}) and response systems (\ref{28}), (\ref{30}).}
\end{figure}
\begin{figure}[htbp] 
   \centering
   \includegraphics[width=\linewidth,height=1.5in]{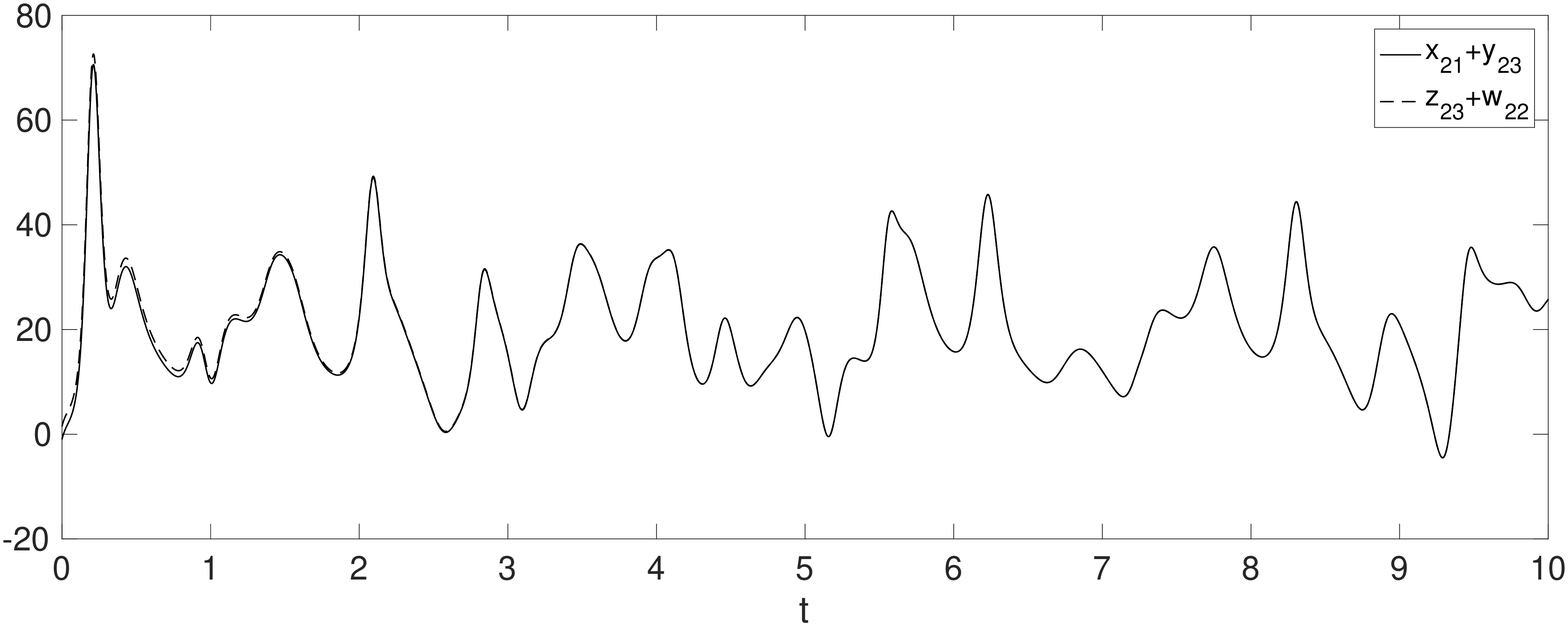} 
   \caption{Response for states $x_{21}+y_{23}$ and $z_{23}+w_{22}$ for drive systems (\ref{24}), (\ref{26}) and response systems (\ref{28}), (\ref{30}).}
\end{figure}
\begin{figure}[htbp] 
   \centering
   \includegraphics[width=\linewidth,height=1.5in]{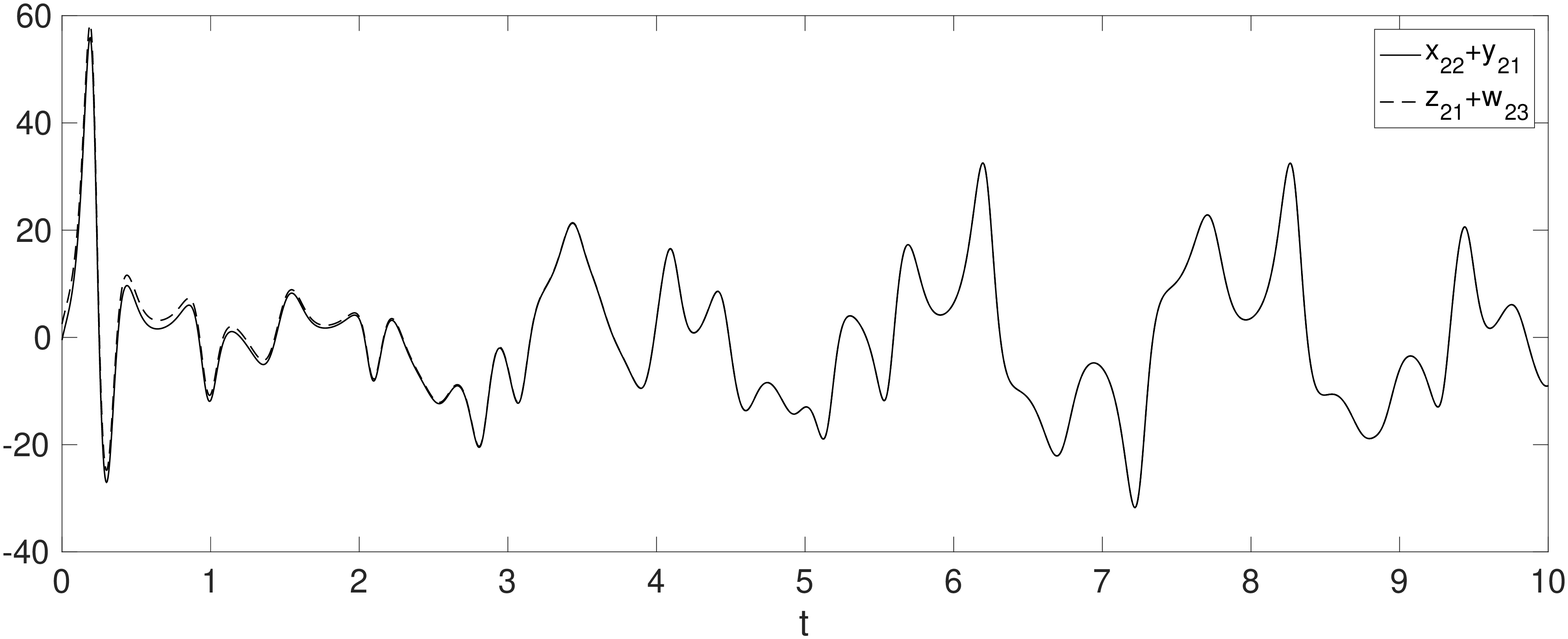} 
   \caption{Response for states $x_{22}+y_{21}$ and $z_{21}+w_{23}$ for drive systems (\ref{24}), (\ref{26}) and response systems (\ref{28}), (\ref{30}).}
\end{figure}
\begin{figure}[htbp] 
   \centering
   \includegraphics[width=\linewidth,height=1.5in]{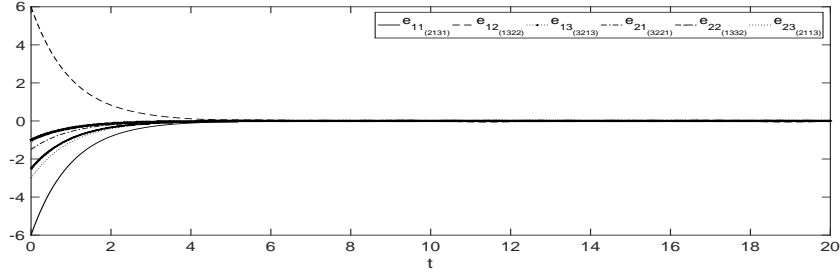} 
   \caption{Time response of synchronization errors}
\end{figure}
\section{Conclusions}
The main purpose of this paper is to propose a novel scheme for synchronization involving eight chaotic systems. The proposed scheme dual combination combination multi switching synchronization achieves synchronization between four chaotic drive systems and four chaotic response systems in a multi switching manner. The main advantages of this scheme can be summarised as

i) The complexity of signal achieved by multiple combination enhances the security of transmitted signal, as the dynamic behaviour of resultant signal is so complex that it becomes very difficult, for the intruder, to separate the information signal from the transmitted signal. Thus in the context of secure communication applications, this scheme may provide improved performance and better resistance.

ii) The concept of multi switching in this scheme further strengthens the anti attack ability of the transmitted signals from drive systems because, for an intruder, determining the correct combination for error space vector is extremely difficult due to large number of possible synchronization directions.

iii) Several other synchronization schemes such as multi switching combination combination  synchronization, dual combination multi switching synchronization, multi switching combination synchronization, dual multi switching projective synchronization, multi switching projective synchronization are obtained as special cases of dual combination combination multi switching synchronization scheme.

iv) The proposed scheme theoretically guarantees a good control performance. 

v) The significant outcome may form the basis of several future studies.

Using Lyapunov stability theory, sufficient conditions are obtained for achieving dual combintaion combination multi switching synchronization. Numerical simulations has been demonstrated using four Genesio Tesi systems and four Lu systems to show the effectiveness and validity of the method. Using fractional order drive and response systems, or utilising the scheme to implement in secure communication applications, and designing the controllers in the presence of uncertain factors and disturbances in the system are some interesting directions for future work.
\section*{Acknowledgements}
The work of the third author is supported by the Senior Research Fellowship of Council of Scientific and Industrial Research, India(Grant no. 09/045(1319)/2014-EMR-I).

\section*{References}

\bibliography{DCCMS}
\bibliographystyle{elsarticle-num}

\end{document}